%% file: main.tex
\documentclass{amsart}
    \usepackage{amsmath}
    \usepackage{amssymb}
    \usepackage{graphicx}
    \usepackage[colorlinks,hypertexnames=false]{hyperref}
    \hypersetup{
      linkcolor=[rgb]{0,0,0.6},
      citecolor=[rgb]{0, 0.4, 0},
      urlcolor=[rgb]{0.6, 0, 0}
    }
    \usepackage{amsthm}
    \usepackage{algorithm}
    \usepackage{algpseudocode}
    \algtext*{EndWhile}
    \algtext*{EndIf}
    \algtext*{EndFor}
    \usepackage{float}
    \usepackage{chngcntr}
    \usepackage{caption}
    \usepackage{bm}
    \usepackage{mathtools}
    \usepackage{thmtools}
    \usepackage[capitalize,nameinlink]{cleveref}
    \usepackage{babel}
    \usepackage{tikz}
    \usepackage{pgfplots}
    \usepackage{complexity}
    \floatname{algorithm}{Algorithm}
    
    \counterwithin{algorithm}{section}
    \newtheorem{definition}[algorithm]{Definition}
    
    \newtheorem{theorem}[algorithm]{Theorem}
    \newtheorem{claim}[algorithm]{Claim}
    \newtheorem{lemma}[algorithm]{Lemma}
    \newtheorem{proposition}[algorithm]{Proposition}

    \newcommand{\Input}{\textbf{Input:} }
    \newcommand{\Output}{\textbf{Output:} }
    \newcommand{\Z}{\mathbb{Z}}
    \newcommand{\F}{\mathbb{F}}
    \newcommand{\N}{\mathbb{N}}
    
    \numberwithin{equation}{section}
    \newcommand{\ord}{\mathrm{ord}}

\begin{document}
        \title{On Deterministically Finding an Element of High Order Modulo a Composite}
        \author{Ziv Oznovich}
        \author{Ben Lee Volk}
        \thanks{Efi Arazi School of Computer Science, Reichman University, Israel. Emails: \{\texttt{oznovich.ziv@gmail.com, benleevolk@gmail.com}\}. The research leading to these results has received funding from the Israel Science Foundation (grant number 843/23).}

        \input{abstract}

        \maketitle
        \markboth{\MakeUppercase{\shortauthors}}{\MakeUppercase{Finding an Element of High Order}}
        \section{Introduction}
        \label{sec:introduction}
        \input{introduction}
        \section{Preliminaries}
        \label{sec:preliminaries}
        \input{preliminaries}
        \section{Factoring \(N\) Given the Residue of a Prime Divisor}
        \label{sec:factoring-N-given-residue}
        \input{factoring-given-a-residue}

        \section{The Main Algorithm}
        \label{sec:main-algorithm}
        \input{main-algorithm}

        \section{Open Problems}
        \label{sec:open}
        \input{open}
        \bibliographystyle{alpha}
        \bibliography{refs}
    \end{document}

%% file: abstract.tex
\begin{abstract}
    \label{sec:abstract}
    We give a deterministic algorithm that, given a composite number $N$ and a target order $D \ge N^{1/6}$, runs in
    time $D^{1/2+o(1)}$ and finds either an element $a \in \Z_N^*$ of multiplicative order at least $D$, or a nontrivial factor
    of $N$.
    Our algorithm improves upon an algorithm of Hittmeir~\cite{hittmeir2018}, who designed a similar algorithm
    under the stronger assumption $D \ge N^{2/5}$.
    Hittmeir's algorithm played a crucial role in the recent breakthrough deterministic integer factorization algorithms of
    Hittmeir and Harvey \cite{hittmeir2021, Harvey21, HarveyHittmeir22}.
    When $N$ is assumed to have an $r$-power divisor with $r\ge 2$, our algorithm provides the same guarantees
    assuming $D \ge N^{1/6r}$.
\end{abstract}

%% file: introduction.tex
While it is generally believed that every randomized algorithm can be simulated by a deterministic algorithm with at
most a polynomial slowdown (and indeed, such a statement is implied by widely believed complexity theoretic
conjectures \cite{NW94, IW97,DMOZ22}), obtaining unconditional derandomization results is still of great interest in
complexity theory and algorithm design.
Computational number theory offers a host of problems in which the best known randomized algorithms significantly
outperform the best known deterministic algorithms, with one prime example being the problem of integer factorization:
the best known randomized algorithms for factoring an integer $N$ with rigorously proved guarantees run in time
$\exp(\tilde{O}(\sqrt{\log N}))$\footnote{Heuristically, the general number field sieve algorithm even runs in time
$\exp(\tilde{O}((\log N)^{1/3}))$, but in this work we focus on algorithms with provable worst case complexity
guarantees.} while the best known deterministic algorithms are still strongly exponential.
For decades, essentially the best known algorithm was the Strassen-Pollard algorithm \cite{Pollard74, Strassen77} that
runs in time $N^{1/4+o(1)}$, and a large body of work was devoted to improving the $o(1)$ term.
Recently, a series of breakthroughs by Hittmeir and Harvey \cite{hittmeir2021, Harvey21, HarveyHittmeir22} resulted in
a deterministic algorithm that runs in time $N^{1/5+o(1)}$.

The algorithms of Hittmeir and Harvey \cite{hittmeir2021, Harvey21, HarveyHittmeir22} all crucially rely on
deterministically finding an element $a \in \Z_N^*$ with large multiplicative order: 
recall that the \emph{order} of an element \(a \in \Z_N^*\), which we denote by \(\ord_N(a)\), is the smallest positive integer \(k\) such
that \(a^k \equiv 1 \pmod{N}\).
Specifically,
\cite{Harvey21, HarveyHittmeir22} require finding an element whose order is at least roughly $N^{1/5}$.

The problem of finding an element of large order is a natural derandomization problem in and of itself: consider, for example, the case where $N=pq$ for primes $p,q$ such that neither $p$ nor $q$
is very small (using standard arguments, Hittmeir and Harvey \cite{hittmeir2021, Harvey21, HarveyHittmeir22} reduce
the general case to this case).
A random element in $\Z_N^*$ is likely to have a large multiplicative order, but it is not immediately clear how to deterministically and efficiently find
 such an element.

This problem was addressed by Hittmeir \cite{hittmeir2018}, who gave an algorithm that, given a composite integer $N$
and a parameter $D\ge N^{2/5}$, runs in time $D^{1/2+o(1)}$, and either finds an element $a \in \Z_N^*$ of
multiplicative order at least $D$, or finds a nontrivial factor of $N$.
It might seem slightly unsatisfactory that the algorithm is not guaranteed to find an element of large order. But 
 in the context of being used as a subroutine in a factorization algorithm, this algorithm leads to a natural win-win
type argument: either such an element is found, or a nontrivial factor of $N$ is found.
Furthermore, there exist infinitely many integers $N$ such that the order of any element modulo $N$ is at most
$(\log N)^{O(\log \log \log N)}$ \cite{EPS91}, so in the general case, an element of such a large order does not always exist.

In the factorization algorithms in \cite{Harvey21, HarveyHittmeir22}, Hittmeir's algorithm \cite{hittmeir2018} is
applied with $D=N^{2/5}$, so it runs in time $N^{1/5+o(1)}$, barely fast enough so as not to hurt the total running
time of the algorithm (as the other steps in the algorithm also take time $N^{1/5+o(1)}$). If it does not factor $N$, it produces an
element whose order is at least $D=N^{2/5}$, whereas the next steps of the factorization algorithm only require an element of order
at least roughly $N^{1/5}$.

As a final remark, we note that when Hittmeir's algorithm \cite{hittmeir2018} is given as an input a prime number $N$,
it will not necessarily output an element of larger order modulo $N$, and it may just detect the fact that $N$ is
prime (which is, of course, famously decidable in deterministic polynomial time \cite{AKS04}).

When $N=p$ is prime, $\Z_p^*$ is cyclic, and deterministically finding a generator for this group is another very
interesting open problem.
To the best of our knowledge, the best known deterministic algorithm for this problem follows from using the
factorization algorithms of \cite{Harvey21, HarveyHittmeir22} to obtain a factorization of $p-1$ as
$p-1=\prod_{i=1}^m q_i^{e_i}$, and then, using brute-force search, finding $q_i$-th non-residue $a_i$ modulo $p$ for every
$q_i$, and outputting $\prod_{i=1}^m a_i^{(p-1)/q_i^{e_i}}$.
Classical upper bounds on the magnitude of the least $q$-th non-residues \cite{Burgess57, Wang64, Norton71, Trevino15}
imply that the time complexity of the algorithm is dominated by the factorization step, so the total running time is
$p^{1/5+o(1)}$.
In \cite{Grossman15}, Grossman presents a \emph{pseudo-deterministic} algorithm for finding a generator modulo $p$ in time $\exp(\tilde{O}(\log p))$ (a pseudo-deterministic algorithm is a randomized algorithm that on every input, is guaranteed to output some canonical output with high
probability).

\subsection{Our results}\label{subsec:our-results}

We give a deterministic algorithm for finding an element $a \in \Z_N^*$ of large order that works for a wider range of
parameters than Hittmeir's algorithm \cite{hittmeir2018}.

\begin{theorem} \label{int:main-algorithm}
There exists a deterministic algorithm that, when given as input a composite integer $N\in\mathbb{N}$ and a target order
$D \ge N^{1/6}$, runs in time $D^{1/2+o(1)}$ and outputs either
a nontrivial factor of $N$ or an element $a \in \Z_N^*$ of multiplicative order at least $D$.
\end{theorem}

Our algorithm can be used as-is as a subroutine in the factorization algorithms of Hittmeir and Harvey
\cite{hittmeir2021, Harvey21, HarveyHittmeir22} to produce an element of order at least, say $N^{1/5}$ (which is a
large enough order for the rest of the algorithm's analysis to go through) in time $N^{1/10+o(1)}$.
As mentioned earlier, in the algorithms of \cite{hittmeir2021, Harvey21, HarveyHittmeir22} the running time of this
step is $N^{1/5+o(1)}$.

Like Hittmeir's algorithm \cite{hittmeir2018}, our algorithm might not produce any useful output if
given as input a prime number $N$ (i.e., it might just output the fact that $N$ is prime).

For integers that have a prime factor with multiplicity $r$, we can even obtain an algorithm with the same running time of $D^{1/2+o(1)}$ under the weaker assumption $D \ge N^{1/6r}$.

In a recent independent and concurrent work, Gao, Feng, Hu and Pan~\cite{cryptoeprint:2025/1004} presented a similar
deterministic algorithm for finding an element of large order modulo composite that runs in time $D^{1/2+o(1)}$ under
the stronger assumption $D\geq N^{1/4r}$.
We remark further on their techniques in the~\hyperref[subsec:proof-technique]{next subsection} and in~\cref{sec:factoring-N-given-residue}.
The paper~\cite{cryptoeprint:2025/1004} further includes improved algorithms (by logarithmic factors) for deterministic
factorization in certain important special cases.

\subsection{Proof Technique}\label{subsec:proof-technique}

We begin with a broad overview of Hittmeir's algorithm \cite[Algorithm 6.2]{hittmeir2018}.

In a high level overview, one can think of Hittmeir's algorithm as performing a brute force search for an element of
high order, starting from $a=2,3,\ldots$ Given a specific element $a \in \Z_N^*$, one can deterministically check in time
$O(\sqrt{D}) \cdot \log^{O(1)}N$ whether the order of $a$ is at least $D$, using a simple variant on Shanks's~\cite{Shanks}
famous Baby-Step Giant-Step method for solving discrete logarithm: indeed, one can compute (and store) the list of
elements $a^j$ for $0 \le j\le \lceil \sqrt{D} \rceil -1$ (these are the ``baby steps'') and the list of elements $a^{-i\lceil \sqrt{D} \rceil}$
for $0 \le i\le \lceil \sqrt{D} \rceil -1$ (these are the ``giant steps''), sort the lists and check for a collision in linear time (in
the lengths of the lists).
If indeed $a^k \equiv 1 \bmod N$ for some $k \le D$, then by writing $k=i \lceil \sqrt{D} \rceil + j$ for $0 \le i,j\le \lceil \sqrt{D} \rceil -1$, we
are guaranteed to find a collision, and in fact in this case we also know the exact order of $a$.
If a collision is not found, then the order of $a$ is greater than $D$ (see \cref{thm:order-calculation-alg} for a
small improvement that shaves off a $\sqrt{\log \log D}$ factor).

Going back to the problem of finding an element of large order, Hittmeir's algorithm maintains a variable $M$ that
equals the least common multiple of the orders of the elements checked so far.
It then does a brute force search for the next element $a$ whose order is not divisible by $M$.
If the order of $a$ is at least $D$, we are done.
Otherwise, we can try to use the fact that the order of $a$ modulo $N$ is small to try to obtain information about the
factors of $N$, and update the variable $M$ as the least common multiple of its previous value and the order of $a$.

The value of $M$ grows by a factor of at least 2 in each iteration.
After $O(\log D)$ iterations, we are at the situation where unless that algorithm already terminated with a correct
output, we know that every prime factor $p$ of $N$ equals $1$ modulo $M$, where $M \ge \sqrt{D}$.
At this point, Hittmeir invokes a variant of Strassen's algorithm \cite{Strassen77} which factors $N$ in time roughly
$\frac{N^{1/4}}{\sqrt{M}}$ (see \cref{subsec:strassens-method} for details).

We improve on this method in several ways.
Our first improvement comes from a better upper bound on the number of consecutive elements $a, a+1,a+2,\ldots \in \Z_N^*$
whose multiplicative order is divisible by $M$. Hittmeir \cite{hittmeir2018} upper bounds this number by $M$, but we
prove an upper bound of $O(\sqrt{M})$, which allows us to increase $M$ up to roughly $D$ (rather than $\sqrt{D}$)
without hurting the running time.

Our second contribution comes from replacing the final factorization step of the algorithm.
The question of finding factors of a number given information about their residue classes has received some attention.
Coppersmith, Howgrave-Graham and Nagaraj \cite{CHGN08} designed a remarkable, efficient and deterministic algorithm
specifically for that problem: given an integer $N$ and a parameter $0\le t<s < N$ such that $s \ge N^{\alpha}$ for
 $\alpha > 1/4$ and $\gcd(t,s)=\gcd(s,N)=1$, their algorithm finds all divisors of $N$ that equal $t$ modulo $s$
in deterministic polynomial time.
Their algorithm even provides an upper bound for the number of such factors that depend only on $\alpha$, although
this is less relevant for our application (such an upper bound was also proved earlier by
Lenstra \cite{lenstra1984divisors}). The algorithm of \cite{CHGN08} is based on Coppersmith's method for finding small
roots of low degree univariate modular polynomials \cite{coppersmith1997}, using the basis reduction algorithm of
Lenstra, Lenstra and Lov\'{a}sz \cite{LLL82}.

At the last step of our algorithm we know that every prime factor $p$ of $N$ equals $1$ modulo $M$ and $M \ge D$.
Assuming $D \ge N^{1/4+\varepsilon}$ we can use the factorization of \cite{CHGN08} to recover the prime factors of $N$.

However, we can even improve the parameters a bit further.
Since the other steps of our algorithm already run in time $D^{1/2+o(1)}$, we can allow the factorization step to run
for longer than polynomial time.
To that end, we adapt the technique of Coppersmith et al.\ \cite{CHGN08} to a recent deterministic algorithm of Harvey
and Hittmeir~\cite{harvey2022deterministic} for finding all \(r\)-th power divisors of an integer \(N\).
Their algorithm also builds on Coppersmith’s method~\cite{coppersmith1997}, which they adapt to exploit the fact that
$N$ has an $r$-th power divisor, enabling them to find all such divisors in time $N^{1/4r}$.

We extend their approach under the assumption that the residue class of the prime divisors of \(N\) modulo some integer
\(s\) is known.
Intuitively, the assumption that $N$ has a prime divisor in a specific residue class modulo $s$ should reduce the
search space by a factor of $s$.
Indeed, we show that this intuition can be realized and present a variant of their algorithm (and in a sense also
of \cite{CHGN08}) that, given $N$ and $s \ge N^{\alpha}$ coprime to $N$, finds all $r$-th power divisors of \(N\) that equal
1 modulo $s$ in time $N^{1/4r - \alpha +o(1)}$.
Quite surprisingly, it seems that the approach of Harvey and Hittmeir~\cite{harvey2022deterministic} gives better
dependence on $\alpha$, even when $r=1$, rather than trying to directly extend \cite{CHGN08}.

We now give an overview of the algorithm of Harvey and Hittmeir~\cite{harvey2022deterministic}, and explain our
extensions.

The main idea in the algorithm (and, for that matter, in the algorithm of Coppersmith et al.\ \cite{CHGN08} as well)
is to search for divisors of $N$ in various intervals, by constructing a low-degree univariate integer polynomial
$h(x)$ such that the integer roots of $h$ correspond to potential divisors of $N$.

To search a given interval of the form $[P-H, P+H]$, Harvey and Hittmeir construct the system of polynomials
\[
    f_i(x) =
    \begin{cases}
        N^{m - \lfloor i/r \rfloor}(P + x)^i, & 0 \le i < rm, \\
        (P + x)^i, & rm \le i < d.
    \end{cases}
\]
(where \(m\) and \(d\) are parameters that are later set to optimize the running time of the algorithm).
These polynomials are designed so that they vanish modulo \(p^{rm}\) at \(x_0 = p - P\), assuming \(p^r \mid N\).

The LLL algorithm~\cite{LLL82} is applied to find a short vector in the lattice spanned by those polynomials: that
vector corresponds to a polynomial \(h(x)\) with small coefficients that vanishes at \(x_0 = p - P\) modulo $p^{rm}$.
Since $|x_0| \le H$, then assuming $H$ is small enough, $|h(x_0)| < p^{rm}$, which implies that $x_0$ is a root of $h$
over the integers.

Harvey and Hittmeir then cover the entire search range $[1,N^{1/r}]$ by such shorter intervals, and since the above
procedure runs in polynomial time, the dominant contribution to the algorithm’s running time arises from the number
of subintervals \([P - H, P + H]\) that must be processed.

In our case, it is known that the residue class of the prime factors of \(N\) modulo some integer \(s\) is 1.
Namely, we know that \(p \equiv 1 \pmod{s}\) for all prime factors \(p\) of \(N\).
This reduces the search space by a factor of \(s\), making it a natural goal to modify the algorithm 
 accordingly and aim for an improvement by a factor of \(s\) in running time.

In Harvey and Hittmeir's algorithm, \cite{harvey2022deterministic}, the basic polynomial of interest is \( f(x)=x+P\), where
\(P\) is the center of the current search interval \([P-H, P+H]\).
The two important properties of this polynomial are:
\begin{itemize}
    \item It is monic.
    This property is used to bound the lattice determinant, which is crucial for the
    efficiency of the algorithm, as the maximal possible size of $H$ depends on this determinant, and we would like to
    maximize $H$ in order to reduce the number of search intervals.
    \item If \(p\) is an \emph{r}-th power divisor of \(N\) ``close'' to $P$, then \(x_0=p-P\), is a ``small'' root of $f(x)$ modulo $p$. 
    This fact remains true for powers and multiples of $f(x)$, and integer linear combinations thereof, hence enabling them (using the LLL basic reduction algorithm) to find \(p\).
\end{itemize}

The key idea of our variant is to replace the polynomial \(P + x\) with a new polynomial that, in some sense, retains
all the useful properties of the original construction.
The fact that the residue class of \(p\) modulo \(s\) is known allows us to
reduce the size of the root we need to recover, from $H$ to roughly $H/s$.
This enables us to increase the size of the interval in which we perform a single search by a factor of $s$, which, in
turn, reduces the number of subintervals required by a factor of $s$, and improves the overall running time.

We remark that our algorithm, with straightforward modifications, works equally well for finding divisors in other residue
classes modulo $s$ (i.e., divisors that equal $t$ modulo $s$ for some $0 \le t \le s-1$). Since this is less useful
in our context, we omit the details.

A similar algorithm with the same running time that uses very similar techniques was recently independently obtained
by~\cite{cryptoeprint:2025/1004}, in a similar context.

\subsection{Organization}
\label{subsec:organization}

In \cref{sec:preliminaries}, we cite relevant prior results and prove some preliminary facts that are useful for our main
algorithm.
In \cref{sec:factoring-N-given-residue}, we give the algorithm for finding factors of \(N\)
given their residue class modulo some integer \(s\).
The main algorithm is given in \cref{sec:main-algorithm}.
We conclude with some open problems in \cref{sec:open}.

%% file: preliminaries.tex
In this section, we cite and prove some preliminary results.
Sections \ref{subsec:order-calculation-algorithm}, \ref{subsec:strassens-method},
\ref{subsec:order-of-all-prime-divisors} and \ref{subsec:lll} summarize important subroutines we use in our algorithms.
In \cref{subsec:number-of-consecutive-roots} we prove some new results that we rely on for our improvement.

\subsection{Order Calculation Algorithm}
\label{subsec:order-calculation-algorithm}

As discussed in \cref{sec:introduction}, a well known theorem of Shanks allows one to test, for each integer $a$,
whether the order of $a$ modulo $N$ exceeds $D$, and if it does not, to compute it exactly.
A small improvement over Shanks' method was given by \cite[Algorithm 4.1]{sutherland2007order} (see also
\cite[Theorem 6.1]{hittmeir2018}).
\begin{theorem}[\cite{sutherland2007order}]
    \label{thm:order-calculation-alg}
    There exists an algorithm that on input $D,N \in \N$ such that $D \le N$ and $a \in \Z_N^*$, runs in time
    $O\left(\frac{D^{1/2}}{\sqrt{\log \log D}} \cdot \log^2 N \right)$,
    and outputs:
    \begin{itemize}
        \item $\ord_N(a)$ if $\ord_N(a) \leq D$
        \item ``\( \ord_N(a) > D \)'' otherwise.
    \end{itemize}
\end{theorem}
\subsection{Strassen's Method}
\label{subsec:strassens-method}

As part of the main algorithm we wish to check if \(N\) has relatively ``small'' prime divisors, and if so, find them.
Also, when a ``small'' order element \(a_i\) is found, we wish to factor \(\operatorname{ord}_N(a_i)\).
For these purposes we use the method of Pollard and Strassen \cite{Pollard74, Strassen77} for finding small factors
of \(N\).
(See also \cite[Proposition 2.5]{Harvey21} for the exact statement on the running time).
\begin{theorem}[\cite{Pollard74, Strassen77}]
    \label{thm:strassens-method}
    There exists an algorithm that on input $L,N \in \N$ such that $L \le N$, tests if \( N \) has a prime divisor
    \( p \leq L \), and if so, finds the smallest such divisor, and runs in time \( O(L^{1/2}\log^3N) \).
\end{theorem}

Note that by setting \(L \coloneqq \sqrt{N}\), this method can be used to completely factor \(N\), since any composite
\(N\) must have a divisor \(p \leq \sqrt{N}\).

\subsection{Order of All Prime Divisors}
\label{subsec:order-of-all-prime-divisors}

We state an important lemma of Hittmeir (\cite[Lemma 2.3]{hittmeir2018}) that allows us to deduce, in certain
cases, the order of $a$ modulo prime factors of $N$.

\begin{lemma}[\cite{hittmeir2018}]
    \label{lem:order-all-primes}
    Let \( N \in \N \) and $a \in \Z_N^*$. Let \( m := \ord_N (a) \).
    Then
    \[
        \ord_p(a) = m \quad \text{for every prime } p \text{ dividing } N
    \]
    if and only if
    \[
        \gcd(N, a^{m/r} - 1) = 1 \quad \text{for every prime } r \text{ dividing } m.
    \]
\end{lemma}

\subsection{Lattice Basis Reduction}
\label{subsec:lll}
An important tool that we use in \cref{sec:factoring-N-given-residue} is the famous LLL \cite{LLL82}
lattice basis reduction algorithm.

\begin{lemma}[\cite{LLL82}]
\label{lem:LLL}
    Let \( d \geq 1\) and \(B \geq 2\). Given a set of linearly independent vectors \(v_0, \dots,v_{d-1}\in \mathbb{Z}^d\)
    such that \( \lVert v_i \rVert \leq B\) for all \(i\), there exists a deterministic algorithm that runs in time polynomial in \(d\) and \(\log B\)
    and returns a nonzero vector \(w\) in the lattice \(L \coloneqq \operatorname{span}_{\mathbb{Z}}(v_0,\dots,v_{d-1})\) such that
    \[
        \lVert w \rVert \leq 2^{(d-1)/4} \left( \det L \right) ^{1/d}
    \]
\end{lemma}

\subsection{Number of Consecutive Roots}
\label{subsec:number-of-consecutive-roots}
In this section we prove a bound on the number of consecutive integers \(a\) that satisfy the equality
\(a^{\ell} \equiv 1 \pmod{N}\).
\cite{hittmeir2018} obtained a bound of $\ell$ using the observation that if $a^\ell=1 \bmod N$, then the same holds for every
$p$ that divides $N$, and a degree $\ell$ polynomial over the field $\mathbb{F}_p$ has at most $\ell$ roots.
We prove an improved upper bound of $O(\sqrt{\ell})$.
In fact, our proof also uses the fact that $a^\ell$ cannot be equal to $1$ modulo a prime $p$ for too many consecutive
integers.
Moreover, if this equality fails modulo a prime divisor $p$ of $N$, it certainly cannot be true modulo $N$.
Our proof uses the following lemma, due to Forbes, Kayal, Mittal and Saha, on the number of solutions for a set of
polynomial equations of a certain form modulo $p$.

\begin{lemma}[\cite{forbes2011}, Lemma 2.1]
    \label{lem:fkms}
    Let $k,\ell$ be positive integers such that \(k\leq \frac{2}{\sqrt{5}} \cdot \sqrt{\ell}+1\) and $p$ be a prime such
    that $p>2\ell$.
    Let \(S=\{(x+a_i)^\ell-\theta_i\}_{1\leq i\leq k}\) be a system of univariate polynomials, where
    \(\theta_i,a_i\in \mathbb{F}_p\) and the \(a_i\)'s are distinct.
    Then, the number of common roots of the polynomials in $S$ is at most \(2\ell/(k-1)+3\).
\end{lemma}

We now present and prove the main result of \cref{subsec:number-of-consecutive-roots}.

\begin{claim}
    \label{cl:consecutive}
    Let $N,\ell$ be large enough positive integers. Suppose $N$ has a prime factor $p$ such that $p>2\ell$. Let $m=10 \sqrt{\ell}$
    and $A=\{a,a+1,\ldots,a+m\}$ be a set of $m$ consecutive integers.
    Then $A$ contains an element $b$ such that $b^\ell \not\equiv1 \bmod N$.
\end{claim}
\begin{proof}

    By assumption, $N$ has a prime factor $p>2\ell$.

    Assuming \(\ell>25\), we know \(10\sqrt{\ell} < 2\ell < p\), so all elements of $A$ are distinct modulo $p$.
    
    Consider~\cref{lem:fkms} over $\F_p$ with $k = \sqrt{\ell}/2$, $\theta_i=1$ and $a_i= i$ for $1 \le i \le k$.
    Note that indeed $k < \frac{2}{\sqrt{5}}\ell + 1$, so
    by~\cref{lem:fkms}, we can conclude that the polynomial system $S = \{(x+i)^\ell - 1\}_{1 \le i \le k}$ has at most
    \[
        2\ell/(k-1)+3 = 2\ell/\left(\sqrt{\ell}/2-1\right)+3 < 5\sqrt{\ell} \text{ (for large enough $\ell$)}
    \]
    common roots over $\F_p$. 
    
    Let $A' = \{a,a+1,\ldots,a+5\sqrt{\ell}\} \subseteq A$.
    Then $A$ contains an element $b=a+j$ which is \emph{not} a common root of $S$.
    That is, for some $1 \le i \le k$, $(b+i)^\ell \not\equiv 1$ modulo $p$, and therefore $(b+i) \not\equiv 1 \bmod N$.
    Since $b+i=a+j+i$ with $j \le 5\sqrt{\ell}$ and $i \le \sqrt{\ell}/2$,
    we have that $b+i \in A$.
\end{proof}

We note that \cref{cl:consecutive} assumes that $N$ has a prime factor which is greater than $2\ell$.
The reason we will be able to get away with this assumption in our main algorithm is that if $N$ has small prime
factors, we can use \cref{thm:strassens-method} to find them quickly.

%% file: factoring-given-a-residue.tex
In this section, we give an algorithm that finds all prime factors of $N$ that equal $1$ modulo $s$, for
$s \ge N^{\alpha}$ coprime to $N$.
Our algorithm is a modification of Harvey and Hittmeir's algorithm for finding $r$-power divisors of $N$
\cite{harvey2022deterministic}.

\begin{definition}
    Let \(N,r \in \mathbb{N}\).
    An integer \(p > 1\) is called an \emph{\(r\)-th power divisor} of \(N\) if \(p^r \mid N\).
\end{definition}

The main theorem we prove in this section is the following.

\begin{theorem}
\label{thm:factoring-N-given-residue}
    Let \(N,s,r \in \mathbb{N}\) be integers such that
    \begin{equation}
    \label{eq:s-definition}
        s \geq N^{\alpha} \text{ for some } 0 < \alpha \leq \tfrac{1}{4r}, \quad \gcd(N, s) = 1.
    \end{equation}
    Then there exists a deterministic algorithm that finds all \emph{\(r\)}-th power divisors of \(N\) that equal $1$
    modulo $s$, and runs in time
    \[
        N^{1/4r - \alpha + o(1)}.
    \]
\end{theorem}

As stated in~\cref{sec:introduction}, a similar result was independently obtained by Gao, Feng, Hu, and
Pan~\cite{cryptoeprint:2025/1004}.

The most general setting is the one in which we do not assume $N$ has any $r$-power divisors for $r \ge 2$.
In this case, we may set $r=1$ and the running time of the algorithm is  $N^{1/4-\alpha+o(1)}$.
However, since the parameters we obtain are better when $N$ is assumed to have $r$ power divisors for $r \ge 2$, and
since the proof is identical for all values of $r$ (and also perhaps under the influence of the original purpose of
the algorithm in \cite{harvey2022deterministic}), we treat $r$ as a parameter and present the proof for all values of
$r$.

We start with the following claim.

\begin{claim}
    \label{cl:new-polynomial}
    Let \(N, s \in \mathbb{N}\) with \(\gcd(N, s) = 1\). Let $s'$ denote the unique integer in $[1,N-1]$ which is
    equivalent to the inverse of $s$ modulo $N$.\footnote{We choose to use the notation $s'$ rather than $s^{-1}$ in
    order to clearly distinguish between multiplication by $s'$ and division by $s$: these operations are equivalent
    modulo $N$ but distinct over the integers.}
    
    Given $H<P \in \N$, suppose $p_ 0 = sx_0 + 1$ is a divisor of $N$ that lies in the interval $[P-H, P+H]$.
    Let $\tilde{P}$ be the unique integer in $[0,s-1]$ which is equivalent to $P$ modulo $s$, and define
    $g(x) \in \Z[x]$ as
    \[
        g(x) \coloneq x + s' + s'(P - \tilde{P}).
    \]
    where every coefficient is taken modulo $N$ (so that every coefficient of $g$ is an integer in $[0,N-1]$).
    Further, let $x' \coloneq x_0 - \frac{P - \tilde{P}}{s}$.
    Then:
    \begin{enumerate}
  	\item \label{item:p0-divides-g(x')} \(g(x') \equiv 0 \pmod{p_0}\)
	\item \label{item:x'-upper-bound}  $|x'| \le H/s+1$.
    \end{enumerate}
\end{claim}

\begin{proof}
    Since \(P \equiv \tilde{P} \pmod{s}\), we have that \(s \mid (P - \tilde{P})\),
    meaning that \((P-\tilde{P})/s\) is an
    integer, and thus \(x'\) is also an integer.
    We observe that \((P-\tilde{P})/s \equiv s'(P-\tilde{P}) \pmod{N}\),
    and that \(s' p_0=s' (s x_0 + 1) \equiv x_0 + s' \pmod{N}\).
    Therefore,
    \begin{gather*}
        g(x') =
        x_0 - (P-\tilde{P})/s + s' + s'(P-\tilde{P}) \\
        \equiv x_0 + s' \equiv
        s' p_0  \pmod{N}.
    \end{gather*}
    Since $g(x')$ is divisible by $p_0$ modulo $N$, and $p_0$ divides $N$, $g(x')$ is divisible by $p_0$ over the integers.
    This establishes \eqref{item:p0-divides-g(x')}.
    
    As for \eqref{item:x'-upper-bound}, by assumption \(p_0=sx_0+1 \in [P-H, P+H]\), which implies that
    \[
    -H \le sx_0+1 - P \le H
    \]
    Adding $\tilde{P}-1$ to both sides of the equation and dividing by $s$, we get that
    \[
    \frac{-H + \tilde{P} - 1}{s} \le x_0 - \frac{P - \tilde{P}}{s} \le \frac{H + \tilde{P} - 1}{s}
    \]
    Since $0\le \tilde{P} <s$ and $x' = x_0 - \frac{P - \tilde{P}}{s}$, \eqref{item:x'-upper-bound} follows.
\end{proof}

\cref{cl:new-polynomial} is used in order to search for prime divisors of $N$ (that equal 1 modulo $s$) in small
intervals of the form $[P-H, P+H]$.

We remark that \cref{cl:new-polynomial} is the main difference between our algorithm and Harvey and Hittmeir's
algorithm \cite{harvey2022deterministic}.
Harvey and Hittmeir observe that a divisor $p_0 \in [P-H, P+H]$ corresponds to a root of $x+P$ modulo $p_0$ of
absolute value at most $H$ (the root is $p_0-P$).
In our setup, the added assumption that $p_0$ equals $1$ modulo $s$ means that the search space is reduced by a
factor of $s$.
Intuitively, this means that the bound on the root size is reduced by a factor of $s$ and this is indeed what is shown
in \cref{cl:new-polynomial}.
However, as it will soon be clear, we also want the polynomial $g(x)$ to be monic in $x$, which leads to the
definition in \cref{cl:new-polynomial}.

We now construct the set of polynomials that are the basis of the lattice whose short vector will be the polynomial
whose roots we will eventually find.

\begin{claim}
\label{cl:f_i's}
Let $N, s$ be as in \cref{cl:new-polynomial}. Let $r,m,d \in \N$ be arbitrary parameters.
Define the polynomials \( f_0, \dots, f_{d-1} \in \mathbb{Z}[x] \) as
    \[
        f_i(x) \coloneq
        \begin{cases}
            N^{m - \lfloor i/r \rfloor}g(x)^i, & 0 \le i < rm, \\
           g(x)^i, & rm \le i < d.
        \end{cases}
    \]
    where $g$ is as in \cref{cl:new-polynomial}.
    Let $p_0 \in [P-H, P+H]$ be an $r$-power divisor of $N$ that equals 1 modulo $s$, and let $x'$ be as in
    \cref{cl:new-polynomial}.
    Then \(p_0^{rm} \mid f_i(x')\) for all \(0 \leq i < d\).
\end{claim}

\begin{proof}
    If \(0 \leq i < rm\), we have \(f_i(x)=N^{m - \lfloor i/r \rfloor}g(x)^i\). By \cref{cl:new-polynomial},
    $p_0^i \mid g(x')^i$ and $p_0^r \mid N$, which means
    $p_0^{r(m-\lfloor i/r \rfloor)} \mid N^{m - \lfloor i/r \rfloor}$.
    Since $r \cdot \lfloor i/r \rfloor \le i$, the claim follows.
    If $rm \le i < d$, then by \cref{cl:new-polynomial}, $p_0^{rm} \mid g(x')^i$.
\end{proof}

The following theorem states that assuming $H$ is small enough, we may efficiently
search intervals of the form $[P-H, P+H]$ for prime divisors that equal $1$ modulo $s$.
In order to make the technical parameters easier to digest, the reader may assume that the parameters $d$ and $m$ are roughly
logarithmic in $N$ (this will be true in the final setting of the parameters).
The theorem is analogous to Theorem 2.1 in \cite{harvey2022deterministic}, and the proof is very similar.

\begin{theorem}
\label{thm:h-interval-theorem}
    Let \(N,s\in \mathbb{N}\) be integers such that  $s \geq N^{\alpha}$ for some $0 < \alpha < 1/2$, and $\gcd(N,s)=1$.
    Let $r \le \log N$ be a positive integer, $d \in \N$ and $m \le d/r$.
    
    Further, suppose $P,H$ are integers such that  $H < P \leq N^{1/r}$.
    
    Then, assuming
    \begin{equation}
    \label{eq:G-size}
        G^{(d-1)/2} < \frac{1}{d^{1/2}2^{(d-1)/4}}\cdot\frac{(P-H)^{rm}}{N^{rm(m+1)/2d}}
        \text{\;\;\; for } G \coloneqq \lceil H/s \rceil +1,
    \end{equation}
    there exists a deterministic algorithm that finds all integers \(p\) such that \(p^r \mid N\) and
    $p \equiv 1 \pmod s$ in the interval
    \([P-H, P+H]\), that runs in time
    \[
        O(d^{c}\log^{k}N)
    \]
    for some constants \(c\) and \(k\).
\end{theorem}

\begin{proof}
\sloppy
We follow the proof of Theorem 2.1 in \cite{harvey2022deterministic}.
    Define the polynomials \( f_0, \dots, f_{d-1} \in \mathbb{Z}[x] \) as in \cref{cl:f_i's}.
    We associate with each polynomial \(f_i(x)\) the coefficient vector $\bar{v}_i$ of the polynomial $f_i(Gx)$.
    That is, \(\bar{v}_i=(a_{i,0},a_{i,1} G,\ldots,a_{i,d-1} G^{d-1})\) where \(a_{i,j}\) is the coefficient of
    \(x^j\) in \(f_i(x)\).

    Define the matrix \(M\) as the matrix whose rows are the vectors \(\bar{v}_i\) for \(0 \leq i < d\).

    We can now run the LLL \cite{LLL82} basis reduction algorithm from \cref{lem:LLL}  on the lattice spanned by the vectors $\bar{v}_i$
    for $0 \le i < d$.
    This algorithm returns a vector \( \bar{w} \) which is an integral linear combination of the $\bar{v}_i$'s, such
    that
    \[
        \lVert \bar{w} \rVert \leq 2^{(d-1)/4} \det(M)^{1/d}.
    \]
    As $M$ is a lower-triangular matrix, its determinant is the product of the leading coefficients of $f_i(Gx)$,
    and we get, using the same calculation as in \cite[Proposition 2.3]{harvey2022deterministic}, that
    \[
        \det(M) = G^{d(d-1)/2}N^{rm(m+1)/2}.
    \]

    The running time of the LLL basis reduction algorithm is polynomial in the dimension of the lattice \(d\) and the
    logarithm of the norm of the basis vectors (see \cref{subsec:lll}).
    Since the entries of \( \bar{v}_i \) are bounded by \(N^{O(d)}\), we can find a non-zero vector \(\bar{w}\) in
    the lattice such that
    \[
        \lVert \bar{w} \rVert \leq 2^{(d-1)/4} G^{(d-1)/2}N^{rm(m+1)/2d}
    \]
    in time $O(d^{c}\log^{k}N)$
    for some constants \(c\) and \(k\).

    Let \(\bar{w} = (w_0, \ldots, w_{d-1}) \). By Cauchy-Schwarz,
\[            | w_0 | + \dots + | w_{d-1} | \leq
            d^{1/2} \lVert \bar{w} \rVert <
            d^{1/2} 2^{(d-1)/4} G^{(d-1)/2}N^{rm(m+1)/2d}.
            \]
    From $\bar{w}$ create the polynomial
        \[
        h(x) \coloneq w_0 + w_1 (x/G) + \ldots + w_{d-1} (x/G)^{d-1}
    \]
    which is a linear combination of the polynomials \(f_i(x)\) for $0 \le i < d$.
    Denoting $h = \sum_{i=0}^{d-1} h_i x^i$, we have that $h_j = w_j/G^j$ (note that $h_j$ is an integer since $w_j$ is
    a linear combination of the integers $a_{i,j}$'s multiplied by $G^j$).
    It follows that
    \begin{equation}
    \label{eq:coefficients-bound}
    |h_0| + |h_1|G + \cdots +|h_{d-1}|G^{d-1} \le d^{1/2} 2^{(d-1)/4} G^{(d-1)/2}N^{rm(m+1)/2d}.
     \end{equation}

    Suppose now $p_0 \in [P-H, P+H]$ is an $r$-power divisor of $N$ that equals 1 modulo $s$, and $x'$ be as in
    \cref{cl:new-polynomial}.

    By \cref{cl:f_i's}, \(p_0^{rm} \mid f_i(x')\) for all \(0 \leq i < d\), and since \(h(x)\) is a linear combination
    of the \(f_i\)'s, we also have that \(p_0^{rm} \mid h(x')\).

    By \cref{cl:new-polynomial}, $|x'| \le G$, so together with~\eqref{eq:coefficients-bound} and~\eqref{eq:G-size} we
    have that
    \[
        |h(x')| \leq |h_0| + |h_1|G + \dots + |h_{d-1}|G^{d-1} < (P-H)^{rm} \leq p_0^{rm}
    \]
    It follows that that \(x'\) is a root of \(h(x)\) also over the integers.

    Therefore, after computing $h(x)$, the algorithm computes the roots of $h(x)$ over the integers in time polynomial
    in $d$ and $\log N$ (see, e.g., \cite[Theorem 1.2]{harvey2022deterministic}) and for each root, checks whether it
    is indeed an $r$-power divisor of $N$ that equals 1 modulo $s$.
\end{proof}

Next, as in~\cite[Section 3]{harvey2022deterministic}, we consider finding $r$-power divisors
\(T \leq p^r \leq T'\) in larger intervals $[T,T']$.
The method would be to apply~\cref{thm:h-interval-theorem} to subintervals of the form \([P-H, P+H]\) and cover the
entire interval \([T, T']\).
As mentioned in~\cite{harvey2022deterministic}, it turns out that the running time is mainly determined by the number
of subintervals, so we aim to maximize the size of \(H\) while ensuring that the condition
in~\cref{thm:h-interval-theorem} holds, mainly condition (\ref{eq:G-size}).
For that, we require \(G\le\tilde{G}\) for
\begin{equation}
\label{eq:G-tilde}
    \tilde{G} \coloneq \frac{1}{d^{1/(d-1)}2^{1/2}} \cdot \frac{T^{2rm/(d-1)}}{N^{rm(m+1)/d(d-1)}}>0.
\end{equation}

Note that compared to \cite{harvey2022deterministic}, we can take $H$ to be roughly larger by a factor of $s$, which
would mean that the number of intervals is smaller by a factor of $s$.

The following lower bound on $\tilde{G}$ is calculated in~\cite{harvey2022deterministic}:

\begin{lemma}[\cite{harvey2022deterministic}, Lemma 3.1]
    \label{lem:harvey-hittmeir-lemma}
    Let \(N, r, d,\) and \(T\) be positive integers, where \(d \geq 2\) and \(T \leq N^{1/r}\).
    Define
    \begin{equation}
        \label{eq:m}
        m \coloneq \left\lfloor \frac{(d-1)\log T}{\log N} \right\rfloor,
    \end{equation}
    and let \(\tilde{G}\) be as defined in (\ref{eq:G-tilde}).
    Then the following inequality holds:
    \[
        \tilde{G} > \frac{1}{3}N^{\theta^2/r - 1/(d-1)},
    \]
    where
    \begin{equation}
        \label{eq:theta}
        \theta \coloneq \frac{r \log T}{\log N}, \quad \text{so that } T = N^{\theta / r} \text{ and } \theta \in [0, 1].
    \end{equation}
\end{lemma}

The next proposition addresses the detection of \(r\)-th power divisors within the interval \([T, T']\), and is
analogous to~\cite[Proposition 3.2]{harvey2022deterministic}.

\begin{proposition}
    \label{prop:t-interval}
    There is a deterministic algorithm with the following properties:
    Given positive integers \(N, r, s, T,\) and \(T'\) satisfying \eqref{eq:s-definition} and
    \begin{equation}
        \label{eq:T-T'}
        4^{\sqrt{\log N)/r}} \leq T \leq T' \leq N^{1/r},
    \end{equation}
    the algorithm outputs a list of all integers \(p\) within the interval \([T, T']\) that satisfy both \(p^r \mid N\)
    and \(p \equiv 1 \pmod{s}\).

    The total running time of the algorithm is bounded by
    \[
        \left( \frac{T' - T}{T} \cdot N^{\theta(1 - \theta)/r - \alpha} \right) \cdot \polylog(N)
    \]
    where \(\theta\) is defined in~\eqref{eq:theta}.
\end{proposition}
\begin{proof}
    The argument closely follows the proof of~\cite[Proposition 3.2]{harvey2022deterministic}.
    As in ~\cite[Proposition 3.2]{harvey2022deterministic}, \cref{lem:harvey-hittmeir-lemma} shows that
    $\tilde{G}$ (defined in~\cref{eq:G-tilde}) satisfies $\tilde{G} \ge \frac{1}{6}N^{\theta^2/r}$ and
    $G = \lfloor \tilde{G} \rfloor$ satisfies
    \[
            G \geq \tilde{G}/2 > \frac{1}{12}N^{\theta^2/r}.
    \]

    Let $H = (G-1)s$ so that \( G = \lceil H/s \rceil + 1 \). Then $H \ge \frac{G}{2}s$, so that
    \[
            H \ge \frac{s}{24}N^{\theta^2/r}.
    \]
    
    We now apply \cref{thm:h-interval-theorem} with $N,r,d,m,G$ and $H$ as above and as set in
    \cref{lem:harvey-hittmeir-lemma}, with $P=T+H, P=T+3H$ and so on, to cover the entire interval $[T,T']$.

    The total number of intervals that must be examined is
    \[
        \left\lceil \frac{T' - T}{2H} \right\rceil \le \frac{T' - T}{\frac{s}{12}N^{\theta^2/r}} + 1 =
        O \left( \frac{T'-T}{T} N^{\theta(1-\theta)/r - \alpha} \right)
     \]
     where we have used the fact that $s\ge N^{\alpha}$ and $T=N^{\theta/r}$.

    For each of these subintervals, applying~\cref{thm:h-interval-theorem} requires time
    \[
        O(\log^k N),
    \]
    for some constant \( k \in \N \).
\end{proof}

We can now prove \cref{thm:factoring-N-given-residue}.

\begin{proof}[Proof of \cref{thm:factoring-N-given-residue}]

The proof closely follows the argument of Theorem~1.1 in~\cite{harvey2022deterministic}.

We start by searching for $r$-power divisors that equal $1$ modulo $s$ using a brute force search, for all integers up
to $4^{\sqrt{\log N /r}}$.
This takes time $N^{o(1)}$.

Letting $2^k = 4^{\sqrt{\log N /r}}$, we now apply \cref{prop:t-interval} successively on the intervals
$[2^k, 2^{k+1}]$, $[2^{k+1}, 2^{k+2}]$ and so on, until reaching $N^{1/r}$.
Since $k \ge 2\sqrt{\log N /r}$, the assumption of \cref{prop:t-interval} is satisfied, and for each such interval,
$\frac{T'-T}{T} \le 2$ and since $\theta \in [0,1]$, $\theta(1-\theta) \le \frac{1}{4}$.

The number of intervals $[2^k, 2^{k+1}]$ is $O(\log N)$, and therefore the total running time is as stated in the
theorem.
\end{proof}

We remark again that the algorithm in \cref{thm:factoring-N-given-residue} can be slightly tweaked for finding $r$-th power divisors in any residue class $t$ modulo $s$ (with the same running time).
The only required change is to define, in \cref{cl:new-polynomial}, $g(x) = x + s't + s'(P-\tilde{P})$ so that $s'p_0 = s'(sx_0 + t) \equiv x_0 + s't \pmod N$.

%% file: main-algorithm.tex
We now give an algorithm satisfying the claim of \cref{int:main-algorithm}.
As we mentioned, the algorithm is based on Algorithm 6.2 in \cite{hittmeir2018}, while incorporating several
simplifications and improvements.
\begin{algorithm}[H]
    \caption{Finding an Element with Large Order or a Nontrivial Factor of \(N\)}
    \label{alg:algorithm}
    \raggedright
    \Input An integer \(N\in \mathbb{N}\), and \(N^{1/6} \leq D \leq N\)\\
    \Output An element \(a \in \mathbb{Z}_N^*\) with \(\ord_N(a) > D\), or a nontrivial factor of \(N\), or ``$N$ is prime''.
    \noindent\rule{\textwidth}{0.4pt}
    \begin{algorithmic}[1]
        \State Apply~\cref{thm:strassens-method} with \(L \coloneqq 2D\) \label{step:strassen}
        \If{a nontrivial factor $K$ of \(N\) is found}
            \State Return $K$
        \EndIf
        \State Set \(M_1 = 1\) and \(a = 2\)
        \For{\(i = 1, 2, \dots\)} \label{step:main-loop}
            \While{\(a \nmid N\) and \(a^{M_i} \equiv 1 \mod N\)} \label{step:while-a-order-div-Me}
                \State \(a \gets a + 1\)
            \EndWhile
            \If{\(a \mid N\)}
                \State Return \(a\) as a nontrivial factor of \(N\) or, if \(a = N\), ``$N$ is prime''. \label{step:a-div-N}
            \EndIf
            \State Run the algorithm of~\cref{thm:order-calculation-alg} on \( D, N\) and $a$  \label{step:BSGS}
            \If{\(\ord_N(a)\) is not found}
                \State Return \(a\) as element with \(\ord_N(a) > D\) \label{step:large-order-a-found}
            \EndIf
            \State Set \(m_i = \ord_N(a)\) and factor $m_i$ using \cref{thm:strassens-method} (with $L \coloneqq \sqrt{m_i}$) \label{step:factor-me}
            \For{each prime \(q\) dividing \(m_i\)} \label{step:foreach-prime-dividing-me}
                \If{\(\gcd(N, a^{m_i/q} - 1) \neq 1\)} \label{step:compute-gcd}
                    \State Return \(\gcd(N, a^{m_i/q} - 1)\) as a nontrivial factor of \(N\) \label{step:gcd-divides-N}
                \EndIf
            \EndFor
            \State Set \(M_{i+1} \gets \operatorname{lcm}(M_i, m_i)\) \label{step:compute-lcm}
            \If{\(M_{i+1} \geq D\)} \label{step:Me-large}
                \If{\(\gcd(N, M_{i+1}) \neq 1\)} \label{step:compute-gcd-M}
                    \State Return \(\gcd(N, M_{i+1})\) as a nontrivial factor of \(N\) \label{step:gcd-divides-N-2}
                \EndIf
                \State Apply~\cref{thm:factoring-N-given-residue} with \(r\coloneq 1, s\coloneq M_{i+1}\) \label{step:factorization-algorithm}
                \State Return some nontrivial factor of \(N\) or ``$N$ is prime''. \label{step:return-factor}
            \EndIf
            \State \(a \gets a + 1\)
        \EndFor
    \end{algorithmic}
\end{algorithm}
\subsection{High-Level Overview of \cref{alg:algorithm}}
\label{subsec:high-level-overview}

The algorithm performs a brute force search for \(a = 2, 3, \ldots\) until either a large order element is found, or
we obtain enough information that enables us to find a factor of \(N\).
Each \(a\) is found to be either

\begin{itemize}
    \item a divisor of \(N\) (leading to the termination of the algorithm with a nontrivial factor of $N$), or
    \item an element with a large order (leading to the termination of the algorithm with an element of large order), or
    \item an element with a low order.
    In this case, the order of $a$ is used to gather information about
    the prime factors of $N$.
    After enough iterations, this enables us to deduce the residue of the prime factors of \(N\) modulo a large
    enough integer to utilize~\cref{thm:factoring-N-given-residue} and find a prime factor of \(N\).
\end{itemize}

\subsection{Main Changes from~\cite[Algorithm 6.2]{hittmeir2018}}
\label{subsec:main-changes}

By applying \cref{cl:consecutive}, the first \emph{while} loop in step \ref{step:while-a-order-div-Me} takes
\(O(\sqrt{M_i})\) iterations (rather than the \(O(M_i)\) given by \cite{hittmeir2018}).
This allows us to modify the condition in step \ref{step:Me-large} to \(M_i \geq D\) instead of \(M_i \geq \sqrt{D}\) without
impacting the overall running time.

By itself, this argument would have allowed us to take $D \ge N^{1/4}$ and the running time would have remained
$\tilde{O}(\sqrt{D})$.

We get further improvement by using \cref{thm:factoring-N-given-residue} to factor $N$ in step
\ref{step:factorization-algorithm} (\cite{hittmeir2018} uses a variant of \cref{thm:strassens-method} to find a factor
of $N$ is this final step).

\subsection{Proof of~\cref{int:main-algorithm}}

We are now ready to prove that \cref{alg:algorithm} satisfies the statements made in \cref{int:main-algorithm}.

\label{subsec:proof-of-main-algorithm}
\begin{proof}[Proof of \cref{int:main-algorithm}]
    ~\\
    \textbf{\textit{Correctness.}}
    The algorithm terminates at steps \ref{step:a-div-N}, \ref{step:large-order-a-found}, \ref{step:gcd-divides-N}
    and \ref{step:return-factor}.
    We will show that if the algorithm terminates at any of these steps, the output is correct.
    Furthermore, we will demonstrate that the algorithm must terminate.

    Step \ref{step:a-div-N} is straightforward, and step \ref{step:large-order-a-found} follows directly from
    \cref{thm:order-calculation-alg}.

    For step \ref{step:gcd-divides-N} we need to show that if \(\gcd(N, a^{m_i/q} - 1) \neq 1\), then it is a
    nontrivial factor of \(N\), i.e., not \(N\) itself.
    Since \(m_i\) is the order of \(a\) and \(m_i/q < m_i\), we have that \(a^{m_i/q} - 1 \not\equiv 0 \mod N\).

    If we reach step \ref{step:factorization-algorithm}, then in every loop iteration \(i\), steps
    \ref{step:foreach-prime-dividing-me}--\ref{step:gcd-divides-N} did not terminate.
    Using \cref{lem:order-all-primes}, we deduce that \(\ord_p(a) = m_i\) for all primes \(p\)
    dividing \(N\), for all \(i\).
    Because the order of every element modulo $p$ divides \(\varphi(p) = p - 1\), we have that
    \(m_i \mid p - 1\), i.e., \(p - 1 = k_i m_i\) for some \(k_i \in \N\).
    Since \(M_{i+1}\) is the least common multiple of all \(m_i\), there exists some \(k \in \N\) such that
    \(k M_{i+1} = p - 1\), and thus \(p \equiv 1 \mod M_{i+1}\) for all primes \(p\) dividing \(N\).
    Furthermore, we know that \(\gcd(N,M_{i+1})=1\) due to step~\ref{step:compute-gcd-M}.
    Since $M_{i+1} \ge D \ge N^{1/6}$, by~\cref{thm:factoring-N-given-residue}, we can factor \(N\) in step \ref{step:factorization-algorithm}.

    For termination it is enough to show that for all \(i\), \(M_{i+1} > M_i\).
    Assume by contradiction that there exists some \(i\) such that \(M_{i+1} \leq M_i\).
    Since \(M_{i+1} = \operatorname{lcm}(M_i, m_i)\), we have that \(m_i \mid M_i\).
    Since \(m_i\) is the order of \(a\), we have that \(a^{M_i} \equiv 1 \mod N\), contradicting the
    condition in step \ref{step:while-a-order-div-Me}.

    \textbf{\textit{Running time.}}
    Step \ref{step:strassen} runs in time \(D^{1/2+o(1)}\), by~\cref{thm:strassens-method}.

    The number of iterations of the loop in step \ref{step:main-loop} is $O(\log D)$.
    This is because for all \(i\), \(M_{i+1} > M_i\), and \(M_{i+1} = \operatorname{lcm}(M_i, m_i)\), so
    \(M_{i+1} \geq 2 M_i\).

    We will now show that each iteration of the loop in step \ref{step:main-loop} runs in time $D^{1/2+o(1)}$,
    thus completing the proof.

    By applying~\cref{cl:consecutive} with \(\ell \coloneqq D\), we observe that the while loop in step
    \ref{step:while-a-order-div-Me} executes at most \(O(\sqrt{D})\) iterations.
    Note that the condition that \(N\) has a prime factor \(p\) such that \(p > 2D\) is satisfied, otherwise a factor
    of \(N\) would be found in step \ref{step:strassen}.

    Computing the $M_i$-th power of $a$ in step \ref{step:while-a-order-div-Me} can be done in time polynomial in
    $\log D$ and $\log N$ by repeated squaring, as \(M_i\) is bounded by \(D\).\

    Computing if the order of \(a\) is greater than \(D\) in step \ref{step:BSGS} is done in \(\tilde{O}(\sqrt{D})\)
    time, by~\cref{thm:order-calculation-alg}.

    Computing the prime factorization of \(m_i\) in step \ref{step:factor-me} is done in \(\tilde{O}(D^{1/4})\) time,
    as \(m_i < D\).
    The number of distinct prime factors of \(m_i\) in step \ref{step:foreach-prime-dividing-me} is at most
    \(\log m_i\), so this step runs in time polynomial in $\log D$ and $\log N$.
    Computing \(\gcd(N, a^{m_i/q} - 1)\) in step \ref{step:compute-gcd} is done in time polynomial in $\log N$, using
    the Euclidean algorithm.

    Computing the least common multiple of \(M_i\) and \(m_i\) in step \ref{step:compute-lcm} can also be done in time
    polynomial in $\log N$, using the Euclidean algorithm.

    Computing \(\gcd(N, M_{i+1})\) in step \ref{step:compute-gcd-M} is done in time polynomial in $\log N$, using the
    Euclidean algorithm.

    Finally, applying \cref{thm:factoring-N-given-residue} in step \ref{step:factorization-algorithm} with
    \(r\coloneq 1, s\coloneq M_{i+1},\alpha\coloneq 1/6\) is done in time \( N^{1/4-1/6+o(1)} \) time,
    which is $D^{1/2+o(1)}$ since $D \ge N^{1/6}$.

    Therefore, the overall time complexity of the algorithm is $D^{1/2+o(1)}$.
\end{proof}

\subsection{The $r$-power divisor case}

Assuming $N$ has an $r$-power divisor $p$, we can slightly improve on the parameters of \cref{int:main-algorithm}.
Note that in this case, the factorization algorithm of
\cref{thm:factoring-N-given-residue} runs in time $N^{1/4r - \alpha + o(1)}$.
Therefore, we may set $\alpha=1/6r$, and the running time of \cref{alg:algorithm} is $D^{1/2+o(1)}$ even for
$D \ge N^{1/6r}$.
 
\begin{theorem}
\label{thm:r-power}
 There exists a deterministic algorithm that, given as input integers $N,r\in\mathbb{N}$ such that \(N\) has an
 $r$-th power divisor, and a target order $D \ge N^{1/6r}$, runs in time $D^{1/2+o(1)}$ and outputs either
 a nontrivial factor of $N$ or an element $a \in \Z_N^*$ of multiplicative order at least $D$.
\end{theorem}

%% file: open.tex
Our work raises several natural questions that may lead to further progress in deterministic algorithms within
computational number theory.

The first promising direction is to use~\cref{alg:algorithm} as a subroutine in the factorization algorithm
of~\cite{Harvey21}.
One can search for an element \(a \in \mathbb{Z}_N^*\) of order greater than \(N^{1/5}\) in time \(N^{1/10+o(1)}\).
This naturally leads to the question of improving the remaining steps of the factorization algorithm to break the
$1/5$ barrier.

A related question concerns the analog of high order elements in prime fields and more specifically, finding a
generator in the cyclic group \(\mathbb{Z}_p^*\).
As mentioned in \cref{sec:introduction}, to the best of our knowledge, the best known deterministic algorithm for this
problem runs in time $p^{1/5+o(1)}$, and it is an interesting open problem to improve this.

Finally, another longstanding open problem is obtaining improved deterministic algorithms for the discrete logarithm
problem in \(\mathbb{Z}_p^*\).
The classical algorithm of Shanks~\cite{Shanks} solves the problem in \(O(\sqrt{p})\) time using the baby-step
giant-step method, and no significantly faster deterministic algorithm is currently known.
Even a modest improvement over this bound would be a significant progress.